\documentclass[11pt]{article}

\usepackage[light,easyscsl,sfmathbb,nowarning]{kpfonts}
\usepackage{fullpage}
\usepackage{url}
\usepackage[hidelinks]{hyperref}
\usepackage[utf8]{inputenc}
\usepackage[small]{caption}
\usepackage{graphicx}
\usepackage{amsmath}
\usepackage{amsthm}
\usepackage{amsfonts}
\usepackage{booktabs}
\usepackage{amssymb,mathtools}
\usepackage{bm}
\usepackage[pdftex,dvipsnames]{xcolor}
\usepackage{authblk}
\usepackage{xspace}
\usepackage{comment}
\usepackage{array}
\usepackage{multirow}
\usepackage{nicefrac}
\usepackage[noend]{algpseudocode}
\usepackage[numbers,sort]{natbib}

\usepackage[british]{babel}
\usepackage{hyperref}
\hypersetup{
    colorlinks=true,
    linkcolor=olive,
    citecolor=purple
}

\hypersetup{pdfinfo={
  Title={Fine-Grained Liquid Democracy for Cumulative Ballots},
  Author={Matthias K{\"o}ppe, Martin Kouteck{\'y}, Krzysztof Sornat, Nimrod Talmon},
  Keywords={liquid democracy, voting, delegations, cumulative ballots}
}}

\title{Fine-Grained Liquid Democracy for Cumulative Ballots}
  \date{}

\author[1]{Matthias K{\"o}ppe}
\author[2]{Martin Kouteck{\'y}}
\author[3]{Krzysztof Sornat}
\author[4]{Nimrod Talmon}
\affil[1]{\small University of California, Davis, USA\hspace{10pt} \texttt{mkoeppe@math.ucdavis.edu}}
\affil[2]{Charles University, Prague, Czech Republic\hspace{10pt} \texttt{koutecky@iuuk.mff.cuni.cz}}
\affil[3]{IDSIA USI-SUPSI, Lugano, Switzerland\hspace{10pt} \texttt{krzysztof.sornat@idsia.ch}}
\affil[4]{Ben-Gurion University of the Negev, Beer Sheva, Israel\hspace{10pt} \texttt{talmonn@bgu.ac.il}}

\newtheorem{theorem}{Theorem}

\newtheorem{proposition}[theorem]{Proposition}
\newtheorem{lemma}[theorem]{Lemma}

\newtheorem{observation}[theorem]{Observation}
\newtheorem{remark}{Remark}
\newtheorem{example}{Example}

\newcommand{\R}{\mathbb{R}}

\newcommand{\Z}{\mathbb{Z}}

\newcommand{\ve}[1]{\bm{#1}}

\newcommand\ved{{\ve d}}

\newcommand\veu{{\ve u}}
\newcommand\vev{{\ve v}}
\newcommand\vew{{\ve w}}
\newcommand\vex{{\ve x}}
\newcommand\vey{{\ve y}}

\newcommand\vezero{{\ve 0}}

\def\R{\mathbb{R}}
\def\Z{\mathbb{Z}}

\def\calS{\mathcal{S}}

\makeatletter
\def\moverlay{\mathpalette\mov@rlay}
\def\mov@rlay#1#2{\leavevmode\vtop{%
   \baselineskip\z@skip \lineskiplimit-\maxdimen
   \ialign{\hfil$\m@th#1##$\hfil\cr#2\crcr}}}
\newcommand{\charfusion}[3][\mathord]{
    #1{\ifx#1\mathop\vphantom{#2}\fi
        \mathpalette\mov@rlay{#2\cr#3}
      }
    \ifx#1\mathop\expandafter\displaylimits\fi}
\makeatother

\newcommand{\bigcupdot}{\charfusion[\mathop]{\bigcup}{\cdot}}

\DeclareMathOperator*{\polylog}{polylog}

\begin{document}

\maketitle

\begin{abstract}
We investigate efficient ways for the incorporation of liquid democracy into election settings in which voters submit cumulative ballots, i.e., when each voter is assigned a virtual coin that she can then distribute as she wishes among the available election options.
In particular, we are interested in fine-grained liquid democracy, meaning that voters are able to designate a partial coin to a set of election options and delegate the decision on how to further split this partial coin among those election options to another voter of her choice.

The fact that we wish such delegations to be transitive---combined with our aim at fully respecting such delegations---means that inconsistencies and cycles can occur, thus we set to find computationally-efficient ways of resolving voter delegations.
To this aim we develop a theory based fixed-point theorems and mathematical programming techniques and we show that for various variants of definitions regarding how to resolve such transitive delegations, there is always a feasible resolution; and we identify under which conditions such solutions are efficiently computable.
\end{abstract}

\section{Introduction}

We discuss liquid democracy (LD), fine-grained liquid democracy (FGLD), and fine-grained liquid democracy for cumulative ballots (CBs). Then we provide a technical and structural overview of the paper.

\paragraph{Liquid democracy (LD)}
In elections that use \emph{proxy voting}~\cite{miller201924} each voter can choose whether to vote directly by casting her vote or to delegate her vote to a delegate of her choice (who votes on her behalf). In \emph{liquid democracy}~\cite{blum2016liquid} (LD, in short), such delegations are transitive.
LD has attracted attention and popularity partially due to the LiquidFeedback tool for deliberation and voting~\cite{liquidfeedback} that allows for such transitive vote delegations and is in use by the German Pirate Party as well as by other political parties around the globe.

In a way, LD is a middle-ground between direct democracy and representative democracy, as the voters who actively vote---sometimes referred to as \emph{gurus}~\cite{zhang2021power}---act as ad-hoc representatives.
To the voters, LD offers more expressiveness and flexibility,
and several generalizations of LD have been proposed that push these aspects even further:
  e.g., G{\"o}lz et al.~\cite{golz2021fluid} suggested to let voters delegate their vote to several delegates;
  Colley et al.~\cite{colley2021smart} suggested to use general logical directives when specifying delegates.

\paragraph{Fine-grained liquid democracy (FGLD)}
Here we follow a specific enhancement of LD, namely \emph{fine-grained liquid democracy (FGLD)}. Originally suggested by Brill and Talmon~\cite{brill2018pairwise}, the general idea is to allow voters to delegate \emph{parts} of their ballots to different delegates, instead of delegating their ballots as a whole.
Specifically, Brill and Talmon~\cite{brill2018pairwise} studied FGLD for ordinal elections, in which each ballot is a linear order over a set of candidates (e.g., an ordinal election over the set of candidates $C = \{a, b, c\}$ may contain a voter voting $a > c > b$, meaning that the voter ranks $a$ as her best option, $c$ as her second option, and $b$ as her least-preferred option). In this context, instead of allowing each voter to either specify a linear order (i.e., vote directly) or delegate their vote to a delegate of their choice, Brill and Talmon~\cite{brill2018pairwise} suggested the following: 
  for each pair of candidates $a$, $b$, each voter is able to either specify whether she prefers $a$ to $b$ or vice versa, or delegate that decision to another voter of her choice.
While indeed offering greater voter flexibility, this ordinal FGLD scheme may result in non-transitive ballots; e.g., consider a voter deciding $a > b$ but delegating the decision on $\{b, c\}$ to a voter who eventually decides $b > c$ and also delegating the decision on $\{a, c\}$ to a voter who eventually decides $c > a$. The resulting ballot, i.e., $a > b, b > c, c > a$, would be intransitive.
Brill and Talmon~\cite{brill2018pairwise} suggested several algorithmic techniques to deal with such possible violations of ballot transitivity.
Following Brill and Talmon~\cite{brill2018pairwise}, Jain et al.~\cite{jain2021preserving} studied FGLD for {\it Knapsack ballots}, in which the ballot of each voter is a subset $Q$ of candidates from a given set of candidates, each candidate has some cost, and there is a restriction that the total cost of the candidates in $Q$ shall respect some known upper bound (Knapsack ballots are useful for participatory budgeting~\cite{goel2019knapsack}).
Again, the main challenge of allowing such fine-grained delegations is that ballots may end up being inconsistent---in the case of FGLD for Knapsack ballots, after following these delegations ballots may end up violating the total cost upper bound.

\begin{remark}\label{remark:one}
At a high level, the basic issue that has to be resolved when considering LD and its generali\-zations---e.g., FGLD---is that, on one hand, we wish to resolve voter delegations transitively and in a way that is as close as possible to voter intentions, but, on the other hand, we have to satisfy certain structural constraints on the admissible ballots that we arrive to:
  e.g., in ordinal FGLD the set of admissible ballots are those that are transitive; and in Knapsack FGLD the set of admissible ballots are those that respect the global budget limit.
These structural constraints mean that it is not always possible to follow voter delegations such that all voter ballots correspond exactly to voter intentions (e.g., in ordinal FGLD one way to satisfy the constraints is by not considering some delegations at all, thus not fully respecting all voter delegations).
Note that this is also the focus of work on smart voting~\cite{colley2021smart}, which is a further generalizations of LD to arbitrary voter directives.

Below we describe what we mean by FGLD for cumulative ballots and, in particular, how we understand voter intentions as described by their partial ballots and their delegations.
\end{remark}

\paragraph{FGLD for cumulative ballots (CBs)}
In this paper we study FGLD for cumulative ballots (CBs).
A \emph{cumulative ballot} with respect to a set $C$ of $m$ candidates is a division of a unit support among the $m$ candidates (visually, a cumulative ballot corresponds to a division of a virtual divisible coin among the candidates):
  e.g., a cumulative ballot with respect to a set $C = \{a, b, c\}$ may be represented by $[0.4, 0.3, 0.3]$,
  meaning that the voter gives support $0.4$ to $a$, support $0.3$ to $b$, and support $0.3$ to $c$.
CBs are used for different social choice settings~\cite{cole1949legal}, such as for multiwinner elections~\cite{meir2008complexity} and for participatory budgeting~\cite{skowron2020participatory}.
From our viewpoint, CBs are especially fitting to FGLD as they are inherently quite expressive.
Furthermore, note that CBs generalize approval ballots and, to a lesser extent, ordinal ballots. 
To get a glimpse into what FGLD for CBs means, consider the following example.

\begin{example}\label{example:one}
Consider a participatory budgeting~\cite{cabannes2004participatory} instance with the following set of projects:
  $p_1$ is a proposal to build a school for $\$1M$;
  $p_2$ is a proposal to renovate a public university for $\$2M$;
  $p_3$ is a proposal to open a new hospital for $\$3M$;
  and $p_4$ is a proposal to refurbish a birthing center for $\$4M$.
With a standard cumulative ballot, a voter may specify, say, a support of $0.1$ to $p_1$, $0.2$ to $p_2$, $0.3$ to $p_3$, and $0.4$ to $p_4$. With a fine-grained liquid democracy cumulative ballot, however, a voter may specify, say, a support of $0.3$ to the set $\{p_1, p_2\}$ (of education-related projects) and the remainder support of $0.7$ to the set $\{p_3, p_4\}$ (of health-related projects), and delegate the decision regarding the specific division of the $0.3$ support between $p_1$ and $p_2$ to some voter of her choice, as well as the decision regarding the specific division of the $0.7$ support between $p_3$ and $p_4$ to (possibly a different) voter of her choice.
\end{example}

Generally speaking, we wish to allow voters to delegate parts of their cumulative ballots to other voters of their choice, so that they could concentrate on the ``high-level'' decisions (such as the division of the unit support between the education-related projects and the health-related projects in Example~\ref{example:one}) but delegate the ``low-level'' decisions further.
But how shall such fine-grained delegations should be understood?
Consider the following continuation of Example~\ref{example:one}.

\begin{example}\label{example:two}
Say that the voter of Example~\ref{example:one} delegate the decision on how to divide the $0.3$ between $p_1$ and $p_2$ to some voter who assigns $0.2$ support to $p_1$ and $0.4$ support to $p_2$; how shall the support of $0.3$ be divided then? Intuitively, we wish to split the $0.3$ support \emph{proportionally} to how the delegate splits her support between $p_1$ and $p_2$. In this example, this means to assign support $0.1$ to $p_1$ and $0.2$ to $p_2$, as $[0.1, 0.2]$ exactly preserves the support ratio of $[0.2, 0.4]$.
\end{example}

Indeed, throughout the paper we concentrate on resolving fine-grained cumulative delegations in a way that would be as close as possible to such proportional (i.e., ratio-preserving) way as described in Example~\ref{example:two}. 
Defining such proportionality formally turns out to not be a straightforward task, mainly because of the possibility of cyclic delegations and the possibility of delegates assigning $0$ support to some candidates.
Thus, in Section~\ref{section:formal}, after providing the needed notation and formally defining our setting, we describe four natural definitions of such proportionality.
As we observe certain drawbacks of the first three definitions, we continue a further algorithmic and computational exploration only with the fourth definition, which is the most promising.

\paragraph{Technical overview}
Before we dive into our formal treatment, we provide a high-level description of our technical approach, via the following example.

Consider an instance $I$ of FGLD for CBs and imagine a solution $S$ that resolves all voter delegations. 
Now, consider some voter $v$. If, in $S$, the delegations of $v$ are all resolved in an exact proportional way (such as in Example~\ref{example:two}), then, in a way, the intentions of $v$ are fully respected; put differently, $v$ should be ``happy''. If, however, this is not the case, then, if had the ability to do so, $v$ would wish to change her ballot to be more proportional.

The above discussion suggests a game-theoretic point of view on our setting. Taking this perspective, we observe that the existence of a ``perfect'' solution (i.e., that satisfies proportionality exactly for all voters) is equivalent to the existence of a Nash-equilibrium in this game. 
Slightly more concretely, we consider the ``best-response'' function of each voter:
  this function, given a solution (such as $S$ above), assigns for a voter $v$ their best-response ballot to $S$ (i.e., if $v$'s ballot is not optimal for $v$ with respect to some $S$, then the best-response function would return some other ballot that is optimal for $v$).
We then concentrate on these best-response functions and recall that a Nash-equilibrium corresponds to a solution that is a fixed point for those best-response functions (i.e. all voters do not wish to change their ballots, which means that the best-response functions, at this particular point, return their input for all voters).

Following our formulation of such games and the relation to fixed points we then utilize the theory of fixed points and articulate certain sufficient conditions that, whenever they are satisfied by our best-response functions, guarantee the existence of solutions.
Furthermore, we show that a deep result from logic~\cite{Renegar92} implies an algorithm to find solutions when the instance size is fixed.
We also attempt to prove that the best-response function for a notion of proportionality which we focus on has some favorable structural properties; in our attempt, we instead computationally find counterexamples to these properties (Lemmata~\ref{lem:contraction}, \ref{lem:pseudomono}, and~\ref{lem:nonunique}).
While we focus on a specific definition of proportionality that is relevant for the setting of FGLD for CBs, our techniques are rather general, thus in Section~\ref{section:generalizations} we discuss wide generalizations of our setting in which our existence and tractability meta-theorems hold as well.

\paragraph{Paper structure}
In Section~\ref{section:formal} we provide some useful notation and formally describe the general setting of FGLD for CBs. Then, in Section~\ref{section:proportionality} we describe our four notions of proportionality for that setting and argue that our fourth definition elegantly overcomes the drawbacks of the other definitions.
In Section~\ref{section:meta} we discuss the existence, structure, and computability of solutions for our four notions of proportionality.

\section{Formal Model}\label{section:formal}

We begin with definitions of some useful notation.
For a natural number $n$, we denote the set $\{1,2,\dots,n\}$ by $[n]$.
Typically we use bold font to denote a vector (or a matrix), e.g., $\vey$ or $\vex$.
By convention we use subscripts in order to point out a specific value in a vector (or a matrix), and write this vector in regular font, e.g., $y_v = y(v)$ and $x_{v,c} = x(v,c)$.
For a vector of real numbers $\vey \in \R^n$ by $\|\vey\|_1$ we denote $\ell_1$-norm of $\vey$, i.e., $\|\vey\|_1 = \sum_{i \in [n]} y(i)$;
and by $\|\vey\|_\infty$ we denote $\ell_\infty$-norm of $\vey$, i.e., $\|\vey\|_\infty = \max_{i \in [n]} y(i)$
Additionally, for $\vex \in \R^n$ and $A \subseteq [n]$ we use the shorthand $\vex_A$ for the subvector with indices $A$, i.e., $|\vex_A| = |A|$ and $\forall_{i \in A} \: x_A(i) = x(i)$.
By $\vezero^n \in \Z^n$ we denote a vector of zeros of length $n$.
When writing $\vey = \vezero$ we use the notation $\vezero$ for a vector of zeros of appropriate length, i.e., the length of $\vey$.

{\it An election} with cumulative ballots consists of a set of $m$ candidates $C = \{c_1, \ldots, c_m\}$ and a set of $n$ voters $V = \{v_1, \ldots, v_n\}$ such that voter $v_i$ corresponds to a cumulative ballot that is represented as a vector $\vev_i = [v_i(c_1), \ldots, v_i(c_m)]$, with $v_i(c_j) \geq 0$ and $\sum_{c \in C} v_i(c) = 1$, such that $v_i(c)$ is the fractional support $v_i$ gives to $c$.\footnote{
  The scenario where a voter might not want to split all of their support among candidates, i.e., where we allow $\sum_{c \in C} v_i(C) < 1$, can be modeled by introducing a ``charity'' candidate to which the remainder of each voter's support is given.}

Our model of FGLD for an election with cumulative ballots is that each voter $v$ partitions $C$ into a family of non-empty and disjoint subsets (also called {\it bundles}) $S_{v,1}, \dots, S_{v,|\calS_v|} \subseteq C$, where $\calS_v = \{ S_{v,1}, \dots, S_{v,|\calS_v|} \}$ so $\bigcupdot_{S \in \calS_v} S = C$; and, for each $S \in \calS_v$ in their partition, sets a delegate $\delta(v,S) \in V$.
We call $b_{v,S} \geq 0$ the \emph{budget} for a bundle $S$, and we require that $\sum_{S \in \calS_v} b_{v,S} = 1$.
Moreover, we require that, if $\delta(v,S) = v$, then $|S| = 1$, which expresses the {\it self-delegation} scenario where voter $v$ makes a direct choice about the candidate from $S$.
Further, without loss of generality, we can require that, if $b_{v,S} = 0$ then $\delta(v,S) = v$ (so also $|S|=1$).

Given such a description of the delegations, it is natural to seek a solution\footnote{
  Hereinafter we use $\R^{nm}$ as a representation for $\R^{V \times C}$; this will be useful for some algebraic operations.}
$\vex \in \R^{V \times C} = \R^{nm}$ that satisfies the following conditions:
\begin{enumerate}

\item The vote $\vex_v$ is a cumulative ballot of weight $1$, i.e., $\|\vex_v\|_1 = 1$;
and

\item the vote $\vex_v$ respects the budgets, i.e., $\|\vex_{v,S}\|_1 = b_{v,S}$.

\end{enumerate}
A vector $\vex$ that satisfies these conditions is referred to as a \emph{solution}.
Next, the fact that voter $v$ delegates a bundle $S \subseteq C$ to a delegate $\delta(v,S)$ means that she wants her solution to relate to the solution of $\delta(v,S)$ in some way.

\section{Notions of Proportionality}\label{section:proportionality}

As discussed above, there are many ways in which this relation can be understood. Next we consider four notions of proportionality that showcase different behaviors that can occur.
To explain these notions, we need some further definitions.

A best-response function $f: \R^{nm} \to \R^{nm}$ describes, for each voter, what is their desired solution, given some solution $\vex$.
That is, given a solution $\vex$, the voter $v$ would be satisfied if their vote was $f(\vex)_v = (f(\vex))_v$.
Using this function, we can define the \emph{regret} of the voter $v$ as $\|f(\vex)_v - \vex_v \|_1$, which quantifies the difference between her current solution and her desired solution.
It is possible that, with respect to a solution $\vex$, a voter $v$ would be satisfied with not just one but a number of solutions; in that case, $f(\vex)_v$ would be a set, and $f$ would be a set-valued function (also called a \emph{correspondence}).
We are ready to describe a few specific notions of proportionality using this pattern.

\subsection*{Exact Proportionality (EP)}
A voter $v$ is satisfied with respect to a solution $\vex$ and a bundle $S$ if either $\|\vex_{\delta(v,S), S}\|_1 > 0$ (that is, their delegate gives positive support to $S$) and the ratios of $\vex_{v,S}$ exactly match the ratios of $\vex_{\delta(v,S), S}$, or in the case when $\|\vex_{\delta(v,S), S}\|_1 = 0$ the voter is always satisfied.
Thus, $f$ is a correspondence with $f(\vex)_{v,S} = (\vex_{\delta(v,S),S} / \|\vex_{\delta(v,S),S}\|_1) \cdot b(v,S)$ if $\vex_{\delta(v,S),S} \neq \vezero$, and $f(\vex)_{v,S}$ is
any non-negative vector of length $|S|$ and weight $b(v,S)$ otherwise.

The main problematic aspect of (EP) is the following:
  {\it Zero-support} of the delegate (i.e. $\vex_{\delta(v,S),S} = \vezero$) is an issue for the delegation and, in (EP), it is resolved in an arbitrary way allowing any split of the budget $b(v,S)$ into candidates in $S$.
We can see this in the following example.

\begin{example}\label{example:ep}
Let $V = \{v,u\}$ and $C = \{c_1, c_2, c_3\}$.
Voter $v$ delegates $S = \{c_1, c_2\}$ with budget $1$ to $u$ (hence $x_{v,\{c_3\}} = 0$).
If voter $u$ defines its cumulative ballot as $\veu = [0.001, 0, 0.999]$, then the only solution under (EP) for $v$ is $[1,0,0]$.
On the other hand, if $u$ will change its ballot slightly by just moving $0.001$ support from $c_1$ to $c_3$, then $\veu = [0, 0, 1]$ and $v$ can split its budget arbitrarily among $c_1$ and $c_2$.
Therefore, any $\vex_v = [a, 1-a, 0]$ for $a \in [0,1]$ is solution for $v$, in particular, $[0,1,0]$ is feasible for $\vex_v$; this is completely different than the only solution before $u$ slightly changed its ballot.
\end{example}

\begin{remark}
Example~\ref{example:ep} indeed also highlights that the solution is not robust to small changes in the input; robustness is an important property in other social choice context as well~\cite{robustness}.
\end{remark}

\subsection*{Exact Proportionality with Thresholds (EP-T)}
The behavior of (EP) in the case of zero-support may be seen as far too arbitrary (as in Example~\ref{example:ep}).
Moreover, it might seem unnatural that, in (EP), the voter $v$ only stops demanding an exactly proportional solution when her delegate's support for $S$ drops down to exactly $0$.
Indeed, it may be better to consider that she loses her confidence in her delegate below {\it the confidence threshold} $\epsilon_{v,S} > 0$ and then she uses her default vote.

To define such notion of proportionality we require more information from each voter $v$, i.e., for every bundle $S \in \calS_v$ we require a {\it weight} $w_{v,S} > 0$ (expressing the voter's confidence in the delegate), and a {\it default} vector $\ved_{v,S} \in \R_{\geq 0}^{|S|}$ such that $\|\ved_{v,S}\|_1 = b_{v,S}$ (which can be used when the delegate supports $S$ too weakly).
One natural example of a default vector is an {\it even-split} which is defined as $d_{v,S}(c) = b_{v,S} / |S|$ for every $c \in S$.

In (EP-T) we define a threshold $\epsilon_{v,S} = 1/w_{v,S}$ such that, if $\|\vex_{\delta(v,S), S}\|_1 \geq \epsilon_{v,S}$, then we demand exact proportionality as in (EP), but if $\|\vex_{\delta(v,S), S}\|_1 < \epsilon_{v,S}$, then we demand that $\vex_{v,S} = \ved_{v,S}$, that is, $v$'s action is her default vector.

Unfortunately, it may happen that there is no solution under (EP-T) as the following example shows.
The underlying reason for this non-existence is the discontinuity of the best-response function at $\epsilon_{v,S}$.
\begin{example}\label{example:ep-t}
  We define 2 voters $\{v,u\}$ and 4 candidates $\{c_1, c_2 , c_3, c_4\}$.
  Table~\ref{table:ep-t} shows the delegations data.
\end{example}

\begin{table}[ht]
\centering
\begin{tabular}[t]{ |c|c|c|c|c| } 
  \hline
  $S$                 & $\ved_{v,S}$ & $b_{v,S}$ & $\delta(v,S)$ & $\epsilon_{v,S}$ \\ \hline
  $S_1 = \{c_1,c_2\}$ & $[0.5, 0]$   & $0.5$     & $u$           & $0.8$ \\ \hline 
  $S_2 = \{c_3,c_4\}$ & $[0.5, 0]$   & $0.5$     & $u$           & $0.8$ \\ \hline
\end{tabular}\hfill
\begin{tabular}[t]{ |c|c|c|c|c| } 
  \hline
  $S$                 & $\ved_{u,S}$ & $b_{u,S}$ & $\delta(u,S)$ & $\epsilon_{u,S}$ \\ \hline
  $S_3 = \{c_1,c_4\}$ & $[0, 0.5]$   & $0.5$     & $v$           & $0.7$ \\ \hline 
  $S_4 = \{c_2,c_3\}$ & $[0, 0.5]$   & $0.5$     & $v$           & $0.4$ \\ \hline
\end{tabular}
\caption{Delegations defined by voters $v$ (right table) and $u$ (left table) in Example~\ref{example:ep-t}.}
\label{table:ep-t}
\end{table} 

\begin{proposition}\label{prop:ep-t-no-solution}
  Under (EP-T), there is no solution to the instance of Example~\ref{example:ep-t}.
\end{proposition}

\begin{proof}
Let us assume, by contradiction, that there is a solution $\vex$ under (EP-T).
We consider two cases: either $x_{v,c_1} + x_{v,c_4} \geq 0.7$ or $x_{v,c_1} + x_{v,c_4} < 0.7$.
 
In the case $x_{v,c_1} + x_{v,c_4} \geq 0.7$ we know that $u$ on $S_3$ splits budget $0.5$ into $c_1$ and $c_4$ proportionally to $x_{v,c_1}$ and $x_{v,c_4}$.
Additionally, we know that $x_{v,c_2} + x_{v,c_3} \leq 0.3 < 0.4 = \epsilon_{u,S_4}$ hence $u$ has to use its default on $S_4$, i.e., $x_{u,c_2} = 0$ and $x_{u,c_3} = 0.5$.
Further, we know that $x_{u,c_1} + x_{u,c_4} = b_{u,S_3} = 0.5$, so $x_{u,c_1} \leq 0.5$.
Therefore, $x_{u,c_1} + x_{u,c_2} \leq 0.5 < 0.8 = \epsilon_{v,S_1}$ so $v$ has to use its default on $S_1$, i.e., $x_{v,c_1} = 0.5$ and $x_{v,c_2} = 0$.
It means also that $x_{v,c_4} \leq 0.5$ and from this and the fact that $u$ on $S_3$ splits budget $0.5$ proportionally we obtain $x_{u,c_4} \leq 0.25$.
Hence, we have $x_{u,c_3} + x_{u,c_4} \leq 0.75 < 0.8 = \epsilon_{v,S_2}$ so $v$ on $S_2$ has to use its default: $x_{v,c_3} = 0.5$ and $x_{v,c_4} = 0$.
But this gives $ x_{v,c_1} + x_{v,c_4} = 0.5$ which is in contradiction with the assumption $x_{v,c_1} + x_{v,c_4} \geq 0.7$.
 
Let us consider the other case, i.e., $x_{v,c_1} + x_{v,c_4} < 0.7$.
First of all, we know that $u$ has to use its default on $S_3$, i.e., $x_{u,c_1} = 0$ and $x_{u,c_4} = 0.5$.
It follows that $x_{u,c_1} + x_{u,c_2} \leq 0.5 < 0.8 = \epsilon_{v,S_1}$ hence $v$ has to use its default on $S_1$ so $x_{v,c_1} = 0.5$ and $x_{v,c_2} = 0$.
We have two subcases: either $x_{v,c_2} + x_{v,c_3} \geq 0.4$ or $x_{v,c_2} + x_{v,c_3} < 0.4$.
\begin{itemize}
  \item when $x_{v,c_2} + x_{v,c_3} \geq 0.4$ then $u$ splits budget $0.5$ on $S_4$ proportionally to $x_{v,c_2}$ and $x_{v,c_3}$, but $x_{v,c_2} = 0$ so we have $x_{u,c_2} = 0$ and $x_{u,c_3} = 0.5$.
  \item when $x_{v,c_2} + x_{v,c_3} < 0.4$ then $u$ has to use its default on $S_4$, so $x_{u,c_2} = 0$ and $x_{u,c_3} = 0.5$.
\end{itemize}
Notice that in both subcases we has to have $x_{u,c_2} = 0$ and $x_{u,c_3} = 0.5$.
It follows that $x_{u,c_3} + x_{u,c_4} = 1 \geq 0.8 = \epsilon_{v,S_2}$ so $v$ wants to split budget $0.5$ among $c_3$ and $c_4$ proportionally to $x_{u,c_3}$ and $x_{u,c_4}$ which are equal, hence $x_{v,c_3} = x_{v,c_4} = 0.25$.
But this gives $ x_{v,c_1} + x_{v,c_4} = 0.75$ which is in contradiction with the assumption $x_{v,c_1} + x_{v,c_4} < 0.7$.
This finishes the description of the example.
\end{proof}

\subsection*{Exact Proportionality with Thresholds, Interpolated (EP-TI)}
In (EP-T), there is a sharp ``loss of confidence'' in $v$'s delegate at the threshold $\epsilon_{v,S}$.
This causes that there might be no solution at all for (EP-T) instance (as in Example~\ref{example:ep-t}).
It may be better to require exact proportionality as long as $\|\vex_{\delta(v,S), S}\|_1 \geq \epsilon_{v,S}$, but then gradually transition to $v$'s default $\ved_{v,S}$ instead of switching to it abruptly.
Formally, we define $f(\vex)_{v,S} = (\vex_{\delta(v,S),S} / \|\vex_{\delta(v,S),S}\|_1) \cdot b_{v,S}$ if $\|\vex_{\delta(v,S),S}\|_1 \geq \epsilon_{v,S}$, otherwise 
\begin{align}
  f(\vex)_{v,S} = \frac{\hspace{-4pt}\vex_{\delta(v,S),S} + (\epsilon_{v,S} - \|\vex_{\delta(v,S),S}\|_1) \cdot \ved_{v,S}}{\|\vex_{\delta(v,S),S} + (\epsilon_{v,S} - \|\vex_{\delta(v,S),S}\|_1) \cdot \ved_{v,S}\|_1} \cdot b_{v,S} \enspace .\label{eq:interpolation-def}
\end{align}
This expression goes from the delegate's solution to $v$'s default as the delegate's support goes from $\epsilon_{v,S}$ to zero.
Hence, $f(\vex)_{v,S}$ is continuous which may be desired property.
A solution under (EP-TI) for an instance from Example~\ref{example:ep-t} is:
$\vex_v^* = [0.5, 0.0, 0.423, 0.077], \vex_u^* = [0.39154, 0.0, 0.5, 0.10846]$.
Let us demonstrate correctness of this solution with respect to delegation $\delta(v,S_2)=u$.
We have $\| \vex_{u,S_2}^* \|_1 =  x_{u,c_3}^* + x_{u,c_4}^* = 0.60846$ which is smaller than $\epsilon_{v,S_2} = 0.8$ so the best-response for $v$ is defined as an interpolation.
First we calculate the interpolation but without proper scaling.
\begin{align}
  \vex_{u,S_2}^* + (\epsilon_{v,S_2} - \|\vex_{u,S_2}^*\|_1) \cdot \ved_{v,S_2}
  = &[0.5, 0.10846] + (0.8 - 0.60846) \cdot [0.5, 0]
  = [0.59577, 0.10846] \label{eq:epti-solution}
\end{align}
Then we scale this solution to achieve weight $0.5 = b_{v,S_2}$ using, indeed, definition of function $f$ under (EP-TI):
\begin{align*}
  f(\vex^*)_{v,S_2} \stackrel{\eqref{eq:interpolation-def}}{=}& \frac{\hspace{-4pt}\vex_{u,S_2}^* + (\epsilon_{v,S_2} - \|\vex_{u,S_2}^*\|_1) \cdot \ved_{v,S_2}}{\|\vex_{u,S_2}^* + (\epsilon_{v,S_2} - \|\vex_{u,S_2}^*\|_1) \cdot \ved_{v,S_2}\|_1} \cdot b_{v,S_2} \\
  \stackrel{\eqref{eq:epti-solution}}{=}& \frac{\hspace{-3pt}[0.59577, 0.10846]}{\| [0.59577, 0.10846] \|_1} \cdot 0.5 = [0.423, 0.077]\enspace ,
\end{align*}
which is equal to $\vex_{v,S_2}$ therefore $v$ uses already best-response on $S_2$ in a solution $\vex^*$.
The remaining delegations can be checked analogously.

Because of continuity of $f(\vex)_{v,S}$ under (EP-TI), this proportionality notion may be seen reasonable, however, the fact that the derivative of $f(\vex)_{v,S}$ at $\epsilon_{v,S}$ is not continuous causes 'sharp' changes when changing a vote slightly---it can be noticed in the following example.

\begin{example}\label{example:ep-ti}
Let $V = \{v,u\}$, $C = \{c_1, c_2, c_3\}$, and let $v$ delegate $S = \{c_1, c_2\}$ with budget $1$ to $u$ by defining its weight to be $w_{v,S} = 100$ (high confidence).
The default vote for this delegation is $\ved_{v,S} = [0,1]$.
If voter $u$ define its cumulative ballot as $\veu = [0.015, 0, 0.985]$, then the only solution under (EP-TI) for $v$ is $[1,0,0]$;
indeed, we have $\epsilon_{v,S} = 0.01$ and $\|\vex_{u,S}\|_1 = 0.015$, so in such a case $v$ keeps a support ratio of $\veu_{S} = [0.015, 0]$.
Next, we analyze how small changes to $\veu$ may change the solution.
\begin{itemize}

\item If $u$ slightly changes its ballot by moving a support of $\gamma = 0.005$ from $c_1$ to $c_3$, then $\veu = [0.01, 0, 0.99]$ and the only solution under (EP-TI) for $v$ has not be changed, i.e., it is $[1,0,0]$.

\item On the other hand, if $u$ changes its ballot by moving a support of $2\gamma$ from $c_1$ to $c_3$, then $\veu = [0.005, 0, 0.995]$, and the only solution under (EP-TI) for $v$ uses interpolation and it is now $[0.5,0.5,0]$.
Indeed, we have
\begin{align*}
  f(\vex)_{v,S} \stackrel{\eqref{eq:interpolation-def}}{=}& \frac{\hspace{-4pt}\vex_{u,S} + (\epsilon_{v,S} - \|\vex_{u,S}\|_1) \cdot \ved_{v,S}}{\|\vex_{u,S} + (\epsilon_{v,S} - \|\vex_{u,S}\|_1) \cdot \ved_{v,S}\|_1} \cdot b_{v,S} \\
  =& \frac{\hspace{-4pt}[0.005,0] + (0.01 - 0.005) \cdot [0,1]}{\|[0.005,0] + (0.01 - 0.005) \cdot [0,1]\|_1} \cdot 1 \\
  =& \frac{\hspace{-2pt}[0.005,0.005]}{\|[0.005,0.005]\|_1} = [0.5,0.5]\enspace .
\end{align*}

\end{itemize}

All in all, a change of $\gamma$ does not alter a solution, but a change of $2\gamma$ alters it significantly.

\end{example}

\subsection*{Weighted Convex Combinations (WCC)}
To avoid the behaviour of (EP-TI), as demonstrated in Example~\ref{example:ep-ti}, we consider the following:
  the solution $v$ is always a combination of her default $\ved_{v,S}$ with her delegate's solution $\vex_{\delta(v,S),S}$, and her confidence in $\delta(v,S)$ is expressed by taking her delegate's solution with weight $w_{v,S}$.
That is, the desired solution is exactly $\ved_{v,S} + w_{v,S} \cdot \vex_{\delta(v,S),S}$ scaled appropriately to sum up to $b_{v,S}$, that is,
\begin{align}
  f(\vex)_{v,S} = \frac{\hspace{-4pt}\ved_{v,S} + w_{v,S} \cdot \vex_{\delta(v,S),S}}{\|\ved_{v,S} + w_{v,S} \cdot \vex_{\delta(v,S),S}\|_1} \cdot b_{v,S} \enspace . \label{eq:wcc-def}
\end{align}
Observe that this means that, if the weight is fixed, then the influence of the delegate decreases as their support decreases, and with the support of the delegate fixed, their influence increases as the weight $w_{v,S}$ increases.

Below we apply the (WCC) proportionality notion to the instance from Example~\ref{example:ep-ti}.
First of all, the initial solution (when $\veu = [0.015, 0, 0.985]$) for $v$ is different---it is $[0.6, 0.4, 0]$---because:
\begin{align*}
  f(\vex)_{v,S} \stackrel{\eqref{eq:wcc-def}}{=}& \frac{\hspace{-4pt}\ved_{v,S} + w_{v,S} \cdot \vex_{u,S}}{\|\ved_{v,S} + w_{v,S} \cdot \vex_{u,S}\|_1} \cdot b_{v,S} \\
  =& \frac{\hspace{-3pt}[0,1] + 100 \cdot [0.015, 0]}{\|[0,1] + 100 \cdot [0.015, 0]\|_1} \cdot 1 = [0.6, 0.4] \enspace .
\end{align*}
Note that the default vote has a strong impact on this solution because the support of $u$ for $S$ is small.
Let us further analyze two changes of $u$'s vote, as in Example~\ref{example:ep-ti}:
\begin{itemize}
  \item If $u$ changes its ballot to $[0.01, 0, 0.99]$, then the only solution under (WCC) for $v$ is $[0.5,0.5,0]$.
  \item If $u$ changes its ballot to $[0.005, 0, 0.995]$, then the only solution under (WCC) for $v$ is $[\frac{1}{3},\frac{2}{3},0]$.
\end{itemize}
Note that the solution changes more smoothly in the case of (WCC) than in (EP-TI).

\section{Existence and Structure of Solutions}\label{section:meta}

To discuss the existence and structure of solutions and the complexity of computing them, we need to introduce some notions from fixed-point theory.
Given a function $f: \R^n \to \R^n$, a point $\vex \in \R^n$ is called a \emph{fixed-point of $f$} if $f(\vex) = \vex$.
The next result is of fundamental importance:
\begin{proposition}[{Brouwer's theorem~\cite{le1912abbildungen}}]
Let $f$ be a continuous function from a compact convex set $K$ to itself.
Then $f$ has a fixed-point.
\end{proposition}
We say that a correspondence $f: K \to 2^K$ has a \emph{closed graph} if the set $\{(\vex,\vey) \in K \times K \mid \vey \in f(\vex) \}$ is closed.
A fixed-point of a correspondence $f$ is a point $\vex$ s.t. $\vex \in f(\vex)$.
\begin{proposition}[{Kakutani's theorem~\cite{kakutani1941generalization}}]
  Let $K$ be a non-empty, compact, and convex subset of $\R^n$, and let $f$ have a closed graph and $f(\vex)$ be non-empty and convex for all $\vex \in K$. Then $f$ has a fixed-point.
\end{proposition}

We will use these theorems to show the existence of solutions.
But, before it, we discuss how to compute solutions. First, notice that complexity classes such as P and NP are of no use when a solution is guaranteed to exist.
Thus, we are interested in the hierarchy of classes below and including TFNP (``Total Function Nondeterministic Polynomial''), which is the class of function problems that are guaranteed to have an answer and this answer can be checked in polynomial time.
PPAD is a subclass of TFNP which is known to be complete for Brouwer fixed-points, meaning there is a function $f$ satisfying the conditions of Brouwer's theorem, such that any problem in PPAD can be reduced to finding a fixed-point of $f$.
PPAD is in general thought to be hard, in the sense that no polynomial algorithm for PPAD is assumed to exist.
Before we can state the current state of the art, we have to introduce yet another notion.
Even when a fixed-point $\vex$ is guaranteed to exist, it might not be rational; thus, it is common to turn to discussing approximate fixed-points.
A point $\vex$ is an \emph{$\epsilon$-weak approximate} fixed-point if $\|\vex - f(\vex)\|_\infty \leq \epsilon$.
It is an \emph{$\epsilon$-strong approximate} fixed-point if it is at distance at most $\epsilon$ from some fixed-point $\vex^*$.
Most results, as well as our treatment, focus on weak-approximations.
The currently best result about finding Brouwer weak-approximate fixed-points is this:
\begin{proposition}[\cite{ChenD08}] \label{prop:ppad}
  An $\epsilon$-weak approximate fixed-point of an $M$-Lipschitz continuous function $f: \R^n \to \R^n$ requires $\Theta((1/\epsilon \cdot M)^{n-1})$ queries to be found.
\end{proposition}
$M$ is intuitively an upper bound on the first derivatives of $f$.
This kind of complexity is essentially a polynomial-time approximation scheme in fixed dimension~$n$, but exponential in $n$ otherwise.

As we will see, sometimes the proportionality conditions we use can be described with quadratic or other types of constraints, so results from mathematical programming become relevant.
A \emph{Quadratically constrained quadratic program (QCQP)} is a collection of quadratic constraints $g_i(\vex) \leq 0$ and a quadratic objective function $g_0(\vex)$ to be minimized over all $\vex \in \R^n$.
The decision problem of the \emph{existential theory of the reals} is to find a solution $\vex$ of a formula $\varphi(\vex)$ that is a quantifier-free formula involving equalities and inequalities of real polynomials.
The following effective theorem was proved by Renegar~\cite{Renegar92}:
\begin{proposition}[{\cite{Renegar92}}] \label{prop:renegar}
  An $\epsilon$-strong approximation $\vex_\epsilon$ of a satisfying assignment $\vex \in \R^n$ of a quantifier-free formula $\varphi(\vex)$ involving $m$ polynomial inequalities of maximum degree $d$ and with coefficients of total encoding length $L$ can be found in time $\max\{L, \log(1/\epsilon)\} \cdot \polylog(L) (md)^{O(n)}$.
\end{proposition}

Finally, a simple heuristic to search for a fixed-point is the \emph{simple iteration} heuristic, which constructs a sequence of points $\vex^0, \vex^1, \dots$ by taking $\vex^0$ an arbitrary initial point from the set $K$, and then setting $\vex^i = f(\vex^{i-1})$.
This procedure converges to a fixed-point if the function $f$ is a \emph{contraction}, which means that there is a non-negative constant $q < 1$ such that for every $\vex \in K$, $\|f(\vex) - f(f(\vex))\| \leq q \cdot \|\vex - f(\vex)\|$ under some norm; this is Banach's fixed-point theorem~\cite{banach1922operations}.
(Essentially, the proof follows as the distance between iterations decreases geometrically with the coefficient $q$.)
There are many other heuristics for finding fixed-points, usually in the guise of ``zero-finding'' or ``root-finding'' heuristics, because finding $\vex$ such that $f(\vex) = \vex$ can equivalently be seen as finding a zero of the function $g(\vex) = f(\vex) - \vex$ or $g(\vex) = \|f(\vex) - \vex\|_1$, for example.

We are now ready to see what these results say about our four notions of proportionality.

\subsection*{Exact Proportionality (EP)}

One can verify that $f$ is a correspondence satisfying the conditions of Kakutani's theorem, thus, a solution is always guaranteed to exist.
However, because $f$ is a correspondence and not a function, it is unclear how to define the simple iteration procedure or apply other zero-finding heuristics.
The problem can be defined as a QCQP as follows.

\begin{observation}\label{obs:ep-constraint}
  If a voter $v$ delegates $S \in \calS_v$ to $u$ and $b_{v,S} > 0$, then proportionality under (EP) is equivalent to
  $$\forall c_1,c_2 \in S \quad \left( u(c_2) \neq 0 \implies \frac{v(c_1)}{v(c_2)} = \frac{u(c_1)}{u(c_2)} \right)\enspace .$$
\end{observation}

From this we can derive a quadratic constraint: $v(c_1) \cdot u(c_2) = u(c_1) \cdot v(c_2)$.
Notice that the constraint behaves precisely as required also in the case when $u$ assigns zero support to $S$, because both sides of the equality will be zero and thus all possible values are permissible for $v(c_1)$ and $v(c_2)$.
A quadratic program that models the problem consists of the following variables and constraints.
\begin{itemize}
	\item definition of the variables: for every voter $v \in V$ and candidate $c \in C$ we define a variable $x_{v,c} \in [0,1]$ that is the fractional support $v$ gives to $c$:
	\begin{align*}
		x_{v,c} \in [0,1] &\quad \forall v \in V, c \in C\enspace ,
	\end{align*}
	\item constraint for cumulative ballots: we fix that every voter splits budget $1$ to the candidates:
	\begin{align}
		\sum_{c \in C} x_{v,c} = 1 \quad \forall v \in V\enspace , \label{eq:sum-to-one}
	\end{align}
	\item delegation budget: voter $v$ has to split budget $b_{v,S}$ among candidates in $S$, hence:
	\begin{align}
		\sum_{c \in S} x_{v,c} = b_{v,S} \quad \forall v \in V, \;\forall S \in \calS_v\enspace , \label{eq:delegation}
	\end{align}
	\item (EP) constraint for every delegation of $S$ by $v$ (due to Observation~\ref{obs:ep-constraint}):
	\begin{align}
		x_{v,c_1} \cdot x_{\delta(v,S),c_2} = x_{\delta(v,S),c_1} \cdot x_{v,c_2} \quad \forall v \in V, \;\forall S \in \calS_v, \;\forall c_1,c_2 \in S\enspace . \label{eq:proportionality}
	\end{align}
\end{itemize}
Thus, one can utilize the theorem of Renegar~\cite{Renegar92} and solve this problem in polynomial time if the dimension is fixed (see Theorem~\ref{thm:fixeddim}).
In practice, there are also many QCQP solvers such as IPopt, Knitro, Gurobi, or Baron, that can be used to solve this problem.

\subsection*{Exact Proportionality with Thresholds (EP-T)}

We have already shown in Proposition~\ref{prop:ep-t-no-solution} that solutions might not exist.
Notice that (EP-T) does not fit the conditions of Brouwer's theorem as $f$ is not continuous.
The lesson here is to be cautious when considering discontinuous best-response functions; while discontinuity does not immediately imply non-existence of solutions, it opens the door to it.

\subsection*{Exact Proportionality with Thresholds, Interpolated (EP-TI)}

Compared with (EP), $f$ is now a function, and we can see that it satisfies the conditions of Brouwer's theorem, so a fixed-point is always guaranteed to exist.
Moreover, we could now use the simple iteration heuristic, as well as any of the zero-finding heuristics.
The problem can be solved by Renegar's algorithm for small $n,m$; see Theorem~\ref{thm:fixeddim}.\footnote{It is also possible to use a QCQP formulation similar to that of (EP) augmented with logical disjunctions that can be formulated using $0/1$ variables, enforced as $x \cdot (1 - x) = 0$.}

\subsection*{Weighted Convex Combinations (WCC)}
Finally, for (WCC), we again observe that the best-response function $f$ is continuous and thus Brouwer's theorem guarantees the existence of a fixed-point.
Moreover, $f$ is amenable to simple iteration and other zero-finding heuristics, and it has continuous derivatives, which can be exploited by many heuristics (unlike (EP-TI), which has discontinuous derivatives).
We also note that (WCC) can be modeled as a QCQP: the constraint \eqref{eq:proportionality} can be replaced by
\begin{align}
		x_{v,c} \cdot \|\ved_{v,S} + w_{v,S} \cdot \vex_{\delta(v,S),S}\|_1 = (d_{v,c} + w_{w,S} \cdot x_{\delta(v,S),c}) \cdot b_{v,S} \quad \forall v \in V, \;\forall S \in \calS_v, \;\forall c \in S\enspace , \label{eq:proportionality_wcc}
\end{align}
which is quadratic in $\vex$ (note that the $\|\bullet\|_1$ in the left hand side is a linear expression in terms of $\vex$).
This also implies that the problem is polynomial-time solvable in fixed dimension by the algorithm of Renegar~\cite{Renegar92} (Theorem~\ref{thm:fixeddim}).

It is natural to wonder whether $f$ is a contraction and whether the simple iteration heuristic \emph{always} converges quickly.
This is, however, not the case:
\begin{lemma} \label{lem:contraction}
  There is an instance with $10$ voters, $5$ candidates, all delegation weights equal to $10$, and a solution $\vex \in \R^{nm}$ such that $\|f(\vex) - f(f(\vex))\|_1 > \|\vex - f(\vex)\|_1$, thus $f$ is not a contraction under the $\ell_1$-norm.
\end{lemma}
In fact, an even weaker notion would do.
It is known that if the so-called pseudo-gradient of the regret function defined as $(f(\vex) - \vex)^2$ (with $(\bullet)^2$ defined coordinate-wise) is \emph{pseudo-monotone}, then certain techniques converge quickly.
(For more details see the book~\cite[Section 2.1, Theorem 2.6]{ModernOptBook}.)
However, we computationally disprove that this is the case.
Specifically for (WCC), to disprove pseudo-monotonicity, one has to find a solution $\vey$ and a fixed-point $\vex$ such that the dot-product $(\vey - f(\vey)) \cdot (\vey - \vex)$ is negative.
Thus, we set up a mathematical program with polynomial constraints that encodes that \textbf{a)} $\vey$ is a fixed-point, \textbf{b)} $\vex$ is a point respecting constraints~\eqref{eq:sum-to-one} and~\eqref{eq:delegation}, and \textbf{c)} minimizing $(\vey - f(\vey)) \cdot (\vey - \vex)$.
We terminate when we find such $\vex, \vey$ for which the optimum is negative, because they constitute a counterexample.
\begin{lemma} \label{lem:pseudomono}
  There is an instance with $4$ voters, $5$ candidates, all weights equal to $10$, and with the defaults being the even-split, for which the pseudo-gradient of $(f(\vex) - \vex)^2$ is not pseudo-monotone.
\end{lemma}
It is known that if the pseudo-gradient is pseudo-monotone, then there is a unique solution.
It is thus natural to ask whether there are instances where, under (WCC), there exist non-unique solutions.
Again, we computationally discover such instances:
\begin{lemma} \label{lem:nonunique}
  There is an instance with $10$ voters, $5$ candidates, all weights equal to $10$, and two fixed-points $\vex_1, \vex_2$ such that $\|\vex_1 - \vex_2\|_1 > 15$.
\end{lemma}
For the constructions and specific counterexamples, see the \texttt{*-counterexample.ipynb} Jupyter notebooks at \url{https://github.com/martinkoutecky/fgld-cb}.

On the positive side, since the constraints on solutions satisfying (EP), (EP-TI), and (WCC) can be formulated as logical connections of polynomial inequalities, Proposition~\ref{prop:renegar} gives:
\begin{theorem} \label{thm:fixeddim}
  For fixed $n$ and $m$, an $\epsilon$-strong approximation of a solution $\vex \in \R^{nm}$ of an instance of FGLD for CB with $n$ voters and $m$ candidates satisfying any of (EP), (EP-TI), or (WCC) can be found in time polynomial in $\log(\vew, \ved, 1/\epsilon)$.
\end{theorem}

\begin{proof}
Our goal is to construct a formula $\varphi(\vex)$ describing a solution $\vex$, and then apply Renegar's algorithm (Proposition~\ref{prop:renegar}).
For (EP), consider the QCQP given by constraints~\eqref{eq:sum-to-one}--\eqref{eq:proportionality}.
A formula $\varphi$ expressing that $\vex$ satisfies all of these constraints is simply their conjunction, the number of constraints is bounded by $O(nm)$, the largest degree is $2$, and the largest coefficient is $1$.
For (EP-TI), we can use an implication: if $\|\vex_{\delta(v,S),S}\|_1 \geq \epsilon_{v,S}$, then constraint~\eqref{eq:proportionality} must hold, otherwise a different quadratic constraint given by $f$ must hold.
The number of constraints is again $O(nm)$, but now the encoding length of coefficients depends on the largest weight $w_{v,S}$ and default vector $\ved$.
This is no issue because the encoding length anyway only enters the complexity of Proposition~\ref{prop:renegar} polynomially.
For (WCC), simply consider the QCQP given by~\eqref{eq:sum-to-one}, \eqref{eq:delegation}, and \eqref{eq:proportionality_wcc}.
The estimates are the same as for (EP-TI).
\end{proof}

\section{Generalizations}\label{section:generalizations}

We will now show a major strength of our treatment: it can be widely generalized.
We start with stating under which conditions a solution is guaranteed to exists; the following is essentially a restatement of Brouwer's theorem:

\begin{theorem}\label{thm:meta_exists}
  Let $n$ be the number of voters and $m$ the number of candidates.
  For each $i \in [n]$, let $K_i \subseteq \R^m$ be a convex and closed set of possible votes of voter $v_i$, and let $K = K_1 \times K_2 \times \cdots \times K_n$.
  For each $i \in [n]$, let $f_i: K \to K_i$ be the best-response function of a voter $v_i$, that is, with respect to any $\vex \in K$, if voter $v_i$ chooses action $f_i(\vex)$, then their individual regret is $0$.

  If each $f_i$ is continuous, then there exists a fixed point $\vex \in K$, that is, there exists, for each voter $v_i \in V$, an action $\vex_{v_i}$, such that their regret is $0$.
\end{theorem}

\begin{proof}
The function $f(\vex) = (f_1(\vex), f_2(\vex), \dots, f_n(\vex))$ is continuous and the set $K$ is convex and closed.
The existence of a fixed-point follows by Brouwer's theorem.
\end{proof}

Intuitively, the theorem above states that, if each voter $v$ has a continuous best-response function~$f$ and their regret is $\|f_v(\vex) - \vex_v\|_1$, or equivalently, if their regret $r_v: K \to \R_{\geq 0}$ is continuous and they can unilaterally decrease it to $0$, then a solution is always guaranteed to exist. (One implication of the equivalence is easy; the other direction follows by, given a regret function $r_v$, defining, for each $\vex \in K$, $f_v(\vex)$ to be some action $\vex_v$ which decreases the regret of $v$ to $0$ with respect to $\vex$.)
Let us outline a few settings which can be captured by Theorem~\ref{thm:meta_exists}:
\begin{enumerate}
\item \textbf{Proportionality per bundle.} Each voter $v$ can specify for each bundle $S$ whether they require (EP-TI) or (WCC) for this delegation.
\item \textbf{(WCC) for subcommittees.} A voter $v$ may wish to delegate their decision to a \emph{committee} of delegates: say that $v$ designates $k$ delegates $v_1, \dots, v_k$, each with a weight $w_1, \dots, w_k$, and the best response of $v$ is to take $\ved_{v,S} + \sum_{i=1}^k w_i \vex_{v_i, S}$ and scale it to be of $\ell_1$-norm $b_{v,S}$.
\item \textbf{Large- vs small-scale decisions.} We have focused on the setting where the voter makes a ``large-scale'' decision of how support should be split among bundles of candidates, and delegates the ``small-scale'' decision within each bundle. Theorem~\ref{thm:meta_exists} captures also the setting where the voter specifies support ratios within bundles (e.g., by specifying a non-negative $|S|$-dimensional vector $\ved_S$ with $\|\ved_S\|_1 = 1$ for each bundle $S$), but delegates the decision of how to split the total support among these bundles to a delegate $\delta(v)$.
\item \textbf{Continuous confidence functions.} In (WCC), a voter expresses their confidence in a delegate through the weight $w_{v,S}$. The influence of $\delta(v,S)$ increases with $w_{v,S}$ and decreases with $\|\vex_{\delta(v,S),S}\|_1$. A voter may specify a less straight-forward interaction. Imagine that there is a candidate $c$ which $v$ is strongly in favor of, and will trust a delegate $\delta(v,S)$ to the degree to which $\delta(v,S)$ is also in favor of $c$. As long as the dependence of the confidence of $v$ in $\delta(v,S)$ is continuous, satisfying solutions are guaranteed by Theorem~\ref{thm:meta_exists}.
\item \textbf{Spatial voting.} In spatial voting~\cite{enelow1984spatial}, a voter's ballot is some real $d$-dimensional vector. We may define FGLD in this setting analogously to the previous point---the action of $v$ will be a combination of their default and the solution of their delegate(s) to the degree of the (continuous) confidence of $v$ in $\delta(v,S)$.
\end{enumerate}

Turning to tractability, we have the following meta-theorems:
\begin{theorem} \label{thm:meta_compute}
  Let $n,m$, and, for each $i \in [n]$, $K_i$ and $f_i$, be defined as in Theorem~\ref{thm:meta_exists}.
  Then:
  \begin{enumerate}
    \item If each $f_i$ is $M$-Lipschitz continuous, then an $\epsilon$-weak approximate fixed-point $\vex$ can be found in $((1/\epsilon) M)^{O(n-1)}$ queries of $f$.
    \item If each $f_i$ is continuous and can be expressed by a quantifier free formula $\varphi_i(\vex)$ with at most $p$ polynomials of maximum degree $d$ and maximum coefficient encoding length $L$, then an $\epsilon$-strong approximation of a solution can be found in time polynomial in $L, p, d$ and $1/\epsilon$, if $n$ and $m$ are fixed.
  \end{enumerate}
\end{theorem}
\begin{proof}
The theorem is a straightforward application of Propositions~\ref{prop:ppad} and~\ref{prop:renegar}, respectively.
\end{proof}

\section{Outlook}

We considered the setting of FGLD for CBs and concentrated on the basic issue of how to resolve voter delegations transitively and in a way that is proportional, for suitable definitions of proportionality.

In the context of FGLD for CBs, our results are an important step towards allowing voter expressiveness and flexibility, advancing the state of the art in fine-grained liquid democracy.

In a more general context, we view our theoretical treatment---culminating in our meta-theo\-rems---as an important result that could be used for other settings (such as those briefly discussed in Section~\ref{section:generalizations}) as well.
In particular, our meta-theorems can be used in social choice settings that are continuous in nature; a particularly promising area is that of spatial voting~\cite{enelow1984spatial}.

Besides using our meta-theorems for such continuous social choice settings, an interesting avenue for future research is to develop analogous meta-theorems for discrete settings. This may be possible using fixed-point theorems for discrete functions, and the logic would be, similarly to the continuous setting, to view a given social choice setting as a game, define appropriate regret functions and apply discrete fixed-point theorems.
Even if the conditions of discrete fixed-point theorems could not be satisfied, one can consider the analogue of a mixed Nash equilibrium, where a solution would not be a single action but rather a distribution on player's actions.
Such meta-theorems may be used also to revisit the setting of ordinal FGLD~\cite{brill2018pairwise} and Knapsack FGLD~\cite{jain2021preserving}.

\section*{Acknowledgements}
We would like to thank the anonymous reviewers for their helpful comments.

Martin Kouteck{\'y} was partially supported by Charles University project UNCE/SCI/004 and by the project 22-22997S of GA ČR.
Krzysztof Sornat was partially supported by
the SNSF Grant 200021\_200731/1
and the European Research Council (ERC) under the European Union’s Horizon 2020 research and innovation programme (grant agreement No 101002854).
Nimrod Talmon was supported by
the Israel Science Foundation (ISF; Grant No. 630/19).

\bibliographystyle{plain}
\bibliography{bib}

\end{document}